\documentclass[11pt]{article}
\usepackage{amsmath,amssymb,amsthm,amsfonts}
\usepackage{fullpage}
\usepackage{url}
\usepackage{color}
\usepackage[]{algorithmic}
\usepackage{algorithm}

\DeclareMathAlphabet{\pazocal}{OMS}{zplm}{m}{n}

\newcommand{\Oh}{\mathcal{O}}

\newcommand{\inprod}[1]{\left\langle #1 \right\rangle}

\newtheorem{theorem}{Theorem}

\newtheorem{lemma}{Lemma}

\newtheorem{definition}{Definition}

\newcommand{\proofbelow}{3pt}
\newcommand{\afterproof}{\hfill $\blacksquare$ \par \vspace{\proofbelow}}
\renewenvironment{proof}{\noindent\textbf{Proof.}\,}{\afterproof}

\author{Vasileios Nakos\thanks{Harvard University. \texttt{vasileiosnakos@g.harvard.edu}. Supported in part by ONR grant
N00014-15-1-2388}} 

\title{On Fast Decoding of High-Dimensional Signals from One-Bit Measurements}

\begin{document}

\maketitle

\begin{abstract}
In the problem of one-bit compressed sensing, the goal is to find a $\delta$-close estimation of a $k$-sparse vector $x \in \mathbb{R}^n$ given the signs of the entries of $y = \Phi x$, where $\Phi$ is called the measurement matrix. For the one-bit compressed sensing problem, previous work \cite{Plan-robust,support} achieved $\Theta (\delta^{-2} k \log(n/k))$ and $\tilde{ \Oh} ( \frac{1}{\delta} k \log (n/k))$ measurements, respectively, but the decoding time was $\Omega  ( n k \log (n / k ))$. In this paper, using tools and techniques developed in the context of two-stage group testing and streaming algorithms, we contribute towards the direction of very fast decoding time. We give a variety of schemes for the different versions of one-bit compressed sensing, such as the for-each and for-all version, support recovery; all these have $poly(k, \log n)$ decoding time, which is an exponential improvement over previous work, in terms of the dependence of $n$.\\

\end{abstract}

\section{Introduction}

\subsection{Standard Compressed Sensing}

The compressed sensing framework describes how to reconstruct a vector (signal) $x \in \mathbb{R}^n$ given the linear measurements $ y = \Phi x$ where $\Phi \in \mathbb{R}^{m \times n}$ for some $m \ll n$. This is an undetermined system with $n$ variables and $m$ equations. In many applications, however, such as images, we know that the vector $x$ can be approximated by a $k$-sparse vector in some basis. In this case, the matrix $\Phi$ contains a sufficient amount of information to roughly recover $x$ if $m$ is large enough; in particular, as shown in \cite{Tao, CandesTao}, the signal can be reconstructed exactly from $\Theta(k \log (n/k)) $ measurements when $\Phi$ is a Gaussian matrix. In order to do this, however, one has to solve the non-convex program \[ \mathop{min} \|x\|_0\text{ s.t. }y = \Phi x\]

Fortunately, \cite{chen2001atomic, Tao} shows that we can use Basis Pursuit (BP), which changes the objective to $ \mathop{min} \|x\|_1$, and still recover a descent approximation of $x$. This can be solved using linear programming.   

Compressed sensing, or sparse recovery, has appeared to be a very useful tool in many areas such as analog-to-digital conversion \cite{Baraniuk}, threshold group testing \cite{dam}, Discrete Signal Processing \cite{Donoho}, streaming algorithms \cite{mut} and bioinformatics \cite{bio}. Depending on the application, different parameters are needed to optimize (measurements, decoding time, encoding time, failure probability). 

Often we discriminate between the for-all model (or universal recovery) and the for-each model( non-universal recovery). In the for-all model, a single matrix is picked, which allows reconstruction of all $k$-sparse vectors, whereas in the for-each model the measurements are chosen at random such that, for some error probability $p$, they will contain plenty of information to reconstruct a single vector $x$ with probability at least $1-p$. We note that all aforementioned papers refer to the for-all model.

Moreover, it is desirable to achieve {\em sublinear } decoding time. The state of the art for the for-each model is \cite{fastforeach}; it achieves $k \cdot poly(\log n)$ decoding time with $\Theta(k \log (n/k))$ measurements. The failure probability was improved later in \cite{foreachlowrisk} using a much more complicated scheme. In the for-all model, Porat and Strauss devised a scheme with $\Oh( k \log (n/k) )$ measurements accompanied with the first sublinear decoding procedure running in time $\Oh ( k^{1- \alpha} n ^{\alpha})$, for any constant $\alpha$ \cite{porat2012sublinear}. Later, in \cite{Porat2}, the authors manage to bring the dependence of the approximation $\epsilon$ fact down to the right order of $\epsilon^{-1}$ and achieve runtime $poly(k,\log n)$, when $\epsilon \leq (\frac{\log k}{\log n})^{\gamma}$, for any constant $\gamma$.

\subsection{One-Bit Compressed Sensing} 

In applications many times compressed sensing measurements must be quantized, since the requirement of infinite precision is not realistic: any measurement must be mapped to a small finite value in some universe. In hardware implementations, for example, where quantizers are implemented using comparators to zero \cite{first}, there is need of quantization to one-bit measurements. Comparators are indeed fast, but they are expensive, so there minimizing their usage is really imporant. Moreover, dynamic range issues are a smaller problem in the case of one-bit quantizers. Details and motivation can be found in \cite{first}.

It is clear that quantization increases the complexity of the decoding procedure and, additionally, is irreversible: given $y = \mathrm{sign}(\Phi x)$ it is impossible to get the exact vector back. Previous results inquired the case in which the quantization maps each coordinate to  $\{-1,+1\}$, which means that we learn only the sign of each coordinate. First, it is not obvious whether there is sufficient information to reconstruct a signal given its one-bit measurements. Of course, since we cannot know the length of the signal, nor the exact signal (even if its length were given), the following question remains: assuming that the length of the signal is $1$, can we find another signal that it is close to it in the $\ell_2$ norm?

The problem was first studied in the work of Boufounos and Baraniuk \cite{first}, where the authors suggest recovering the signal $x$ by solving the optimization problem \[ min_x \|x \|_1 \text{ s.t.: } y \odot Ax \geq 0, \|x\|_2 =1, \] where $\odot$ stands for the element-wise product between two vectors. The goal is to find a vector $y$ on the unit sphere such that $\|y - \frac{x}{\|x\|_2} \|_2^2 \leq \delta$. It is clear that this relaxation requires solving a non-convex program, something which Laska et al. \cite{laska} tried to remedy by giving  an optimization algorithm that finds a stationary point of the aforementioned program; both papers, however, do not provide provable guarantees for the number of measurements needed. An alternative formulation was studied in \cite{jacques}, which showed that the number of measurements could be  brought down to $\Oh (\delta^{-1} k \log n)$, but the main obstacle of the non-convex formulation remained. In \cite{plan2013one} Vershyin and Plan gave the first computationally tractable algorithm for the problem of one-bit compressed sensing by designing a compressed sensing scheme that approximately recovers a $k$-sparse vector from $\Oh (\delta^{-5} k \log^2(\frac{n}{k}))$ one-bit measurements via a linear programming relaxation. Their techniques were based on random hyperplane tessellations; the main geometric lemma they needed was that $\Oh (k \log(n/k))$ random hyperplanes partition the set of $k$-sparse vectors with unit norm into cells, each one having small diameter. In \cite{Plan-robust} the same authors improved the number of measurements to $\Oh (\delta^{-2} k \log(\frac{n}{k}) )$  by analyzing a simple convex program. Their results can also be generalized to other sparsity structures, where the crucial quantity that determines the number of measurements  is the gaussian mean-width of the set of all unit vectors having a specific sparsity pattern. Last but not least, they manage to handle gaussian noise and, most importantly, adversarial bit flips, though with a small worsening in the dependence on $\delta$ in their number of measurements. In \cite{support} a two-stage algorithm with $\tilde{\Oh} ( \frac{1}{\delta} k \log (n /k) )$ measurements and $\Oh (nk \log (n/k) + \frac{1}{\delta^5} (k \log (n/k))^5)$ decoding time was proposed. Apart from recovering the vector, other algorithms that recover only the support of the signal have been proposed; see for example \cite{support, gupta}.

In \cite{goyal1998quantized} it is suggested that even if the support of the vector is known, the dependence of the number of measurements on $\delta$ must be at least $\frac{1}{\delta}$. In order to circumvent this, alternative quantization schemes were proposed, with the most common being Sigma-Delta quantization \cite{sigmadelta1, sigmadelta2}. In	\cite{baraniuk2014exponential} Baraniuk, Foucart, Needell, Plan and Wooters  manage to bring the dependence on $\delta$ down to $\log ( \frac{1}{\delta})$ if the quantizer is allowed to be adaptive and the measurements take a special form of threshold signs.

\subsection{Group Testing}

In the group testing problem, we have a large population, which consists of ``items'', with a known number of defectives. The goal is to find the defectives using as few tests as possible, where a test is just a query whether a certain subset of items contains at least one defective. The group testing problem was first studied by Dorfman in \cite{dorfman1943}. There are two types of algorithms for this
problem, namely adaptive and non-adaptive. In the first case, the outcome of
previous tests can be used to determine future tests, whereas in non-adaptive algorithms all tests
are performed at the same time. Group testing has many applications in DNA library screening and detection of patterns in data; more can be found in \cite{pooling}, \cite{angle}. 

Any solution for the group testing problem corresponds to a binary matrix, where the number of rows equals the number of tests and the number of columns equals the cardinality of the population. Given such a matrix $M$ and a vector $x$ indicating the positions of the defectives, we should be able to identify $x$ from $Mx$ , where the addition here corresponds to the addition operation of Boolean algebra.  Since decoding time is important, the brute-force algorithm that iterates over each possible subset in order to recognise the defective set does not suffice. However, one can design matrices such that the naive decoding algorithm, which eliminates items belonging to negative tests and returns all the other items, correctly identifies all defective items \cite{du1999combinatorial}. In literature these matrices are known as $k$-disjunct matrices.

In this paper, we are also interested in the so-called two-stage group testing problem, where two stages are allowed: the first stage recognises a superset of the defectives, and the second stage, which is performed after seeing the results of the first stage, recognizes the exact set of the defectives by querying separately for each one. We refer to $(k,l)$ two-stage group testing as the case when there are $k$ defectives and the superset is allowed to have up to $k+l$ elements. In fact, this is equivalent to the existence of a matrix $M$ such that given $Mx$ one can find a set $S$ with $k+l$ elements such that all defectives are included in $S$. The same naive algorithm, which eliminates all items that belong to a negative test and returns all other items, will be used here. A matrix is called list-disjunct if this algorithm finds a superset of the support with at most $k+l$ elements. The term `list-disjunct' appeared in \cite{indyk}, although it was also studied before in \cite{de2005optimal}, under the name of super-imposed codes, and in \cite{rashad1990random}, under the name of list-decoding super-imposed codes. In \cite{HPA} Ngo, Porat and Rudra give efficient and strongly explicit constructions of matrices that allow two-stage group testing, which are also error-tolerant, in the sense that they can correct $e_0$ false positives and $e_1$ false negatives in sub-linear time and additional $\Theta( e_0 + k\cdot e_1)$ tests. They also prove matching lower bounds for several cases, including the case that $k = \Theta(l)$.

A group testing scheme is a tuple $(M,R)$, where $M$ is a matrix in $\{0,1\}^{m \times n}$ and $R$ a procedure that takes as input $Mx$ and outputs a vector $y$. Depending on the guaratee we want, we will either refer to it as two-stage group testing or (one-stage) group testing. The worst-case running time of the procedure $R$ corresponds to the decoding time of the scheme.

\subsection{Our Results} 

The main goal of our work is to understand under which conditions and which number of measurements sublinear decoding time is possible. The paper is divided into two parts. The first part investigates the for-each version of noisy one-bit compressed sensing and give a scheme that not only is exponentially faster than what is known in the literature, but also outperforms, for some setting of parameters, previous work \cite{Plan-robust,support}. In the worst case, our results have a small overhead in the number of measurements. The second part focuses on fast decoding of noiseless vectors. We first give a near optimal scheme with sublinear decoding time, by connecting the problem with Combinatorial Group Testing. Second, we try to understand how it is possible to achieve a for-all guarantee for one-bit compressed sensing, while still keeping sublinear decoding time. Our techniques also give a scheme for support recovery that outperforms the one in \cite{support}, being exponentially faster; one additional aspect of our scheme, is that it is also computable in polynomial time. 

The $\delta-\ell_2/ \ell_2$ guarantee for one-bit compressed sensing is defined as follows: For a unit vector $x \in \mathbb{R}^n$ we say that a scheme satisfies the $\delta-\ell_2/ \ell_2$ guarantee for one-bit compressed sensing, if the output satisfies \[ \| \hat{x} - x \|_2^2 \leq c \|x_{tail(k)}\| _2^2 + \delta, \]
while $x_{tail(k)}$ is the vector that occurs after zeroing out the biggest $k$ coordinates of $x$ in magnitude and $c$ is some absolute constant.\\
In the support recovery problem, one wants to construct a matrix $\Phi$, such that for all $k$-sparse $x$, one is able to recover the support of the vector $x$, given measurements $y = \mathrm{sign}(\Phi x)$.

We present the results that we have in greater detail. We note that the decoding time of each scheme is $poly(k, \log n)$.
 
\begin{itemize}
\item $\delta-\ell_2/ \ell_2$ for-each one-bit Compressed sensing from $\Oh(k \log n \cdot (\log k + \log \log n) + \delta^{-2}  k )$ measurements.   
\item For-each one-bit Compressed Sensing (noiseless signals) from $\Oh(k \log n + \log_k n \cdot \log \log_k n + \delta^{-2}k )$ measurements. This extends the result of \cite{Plan-robust}, as it manages to also decrease the number of measurements for the for-each version of the problem. We note that for $k = \Omega( \frac{\log n}{\log \log n})$, our result uses less measurements than the one presented in \cite{Plan-robust}. 
\item  For-all one-bit Compressed Sensing (noiseless signals) in $\Oh(k^2 \log n \log \log_k n + \delta^{-2} k \log n)$ measurements. This is the first scheme that allows sublinear decoding time in the for-all model, although the dependence on $k$ is $k^2$.
\item Support recovery from one-bit measurements (noiseless signals) in $\Oh(k^3 \log n )$ measuremnnts. This scheme is not only exponentially faster than the one presented in \cite{support}, but also explicit, in the sense that the matrix can be computed in polynomial time in $n$. in the number of measurements.

An interesting aspect of our results is that in the for-each model, the $\delta$ factor does not need to multiply the $k \log n$ factor, in contrast to the for-all version. 

\end{itemize}

\begin{center}
\begin{table}

 \centering
 \begin{tabular}{| c c c c c |} 
 \hline
 Algorithm & Measurements & Decoding-Time & Model & Noise\\ [0.5ex] 
 \hline\hline
\hline 
\cite{Plan-robust} & $\Oh(\delta ^{-6} k \log (n/k)) $ &  $poly(n)$ & For-all & Type 1\\
\hline
\cite{Plan-robust} & $\Oh(\delta^{-2}k \log (n/k) $ & $poly(n)$ & For-all & No\\
 \hline
  \cite{Plan-robust} & $\Oh(\delta^{-2} k \log (n/k))$ & $poly(n)$ & For-each & Type 2 \\
 \hline
 \cite{support} & $\tilde{\Oh}(\delta^{-1} k \log (n/k))$ & $\Oh(nk \log n) + poly(k, \log n)$ & For-all & No \\
\hline
This paper & $\Oh(k \log n (\log k + \log \log n) + \delta^{-2}k)$ & $poly(k,\log n)$ & For-each & Type 3\\
\hline
 This paper & $\Oh(k \log n + \delta^{-2} k + \log _k n \log \log_k n )$ & $poly(k, \log n) $& For-each & No\\
 \hline
 This paper & $\Oh(k^2 \log n \log \log_k n + \delta^{-2} k \log (n/k) ) $ & $poly(k,\log n) $ & For-all & No \\
 \hline
\end{tabular}
 \caption{Comparison of recovery schemes for one-bit compressed sensing }
\end{table}
\end{center}

We explain the three types of noise appeared in the previous table:

\begin{itemize}

\item Type 1 stands for adversarial bit flips. This means that after receiving $y = \mathrm{sign}(Ax)$, an adversary can flip some of the entries of $y$, and then give it to the decoder. Here, we assume that $x$ is exactly $k$-sparse.
\item Type 2 noise stands for gaussian random noise that is added to the $k$-sparse vector $x$ after the matrix $A$ has been applied to it. This means that $y = \mathrm{sign}(Ax + u)$, where $u \sim \mathcal{N} (0,I)$.
\item Type 3 noise refers to general vectors and is handled by the $\delta-\ell_2/\ell_2$ guarantee. This means that $y = \mathrm{sign}(A(x_{head(k)}+x_{tail(k)}))$, where we can view the term $x_{tail(k)}$ as pre-measurement adversarial noise. 
\end{itemize}

\begin{center}
\begin{table}

 \centering
 \begin{tabular}{| c c c c|} 
 \hline
 Algorithm & Measurements & Decoding time & Model  \\ [0.5ex] 
 \hline\hline
 \hline
  \cite{support} & $\Oh(k^3 \log n )$ & $\Oh(nk \log n)$ & For-all \\
  \hline
 This paper & $\Oh(k^3 \log n  \ )$ & $\Oh(k^3 poly(\log n))$ & For-all \\
 \hline
\end{tabular}
 \caption{Comparison of schemes for support recovery}
\end{table}
\end{center}

The ideas that are used in this paper to get sublinear decoding time are based on ideas that appeared in \cite{HPA}, as well as the dyadic trick, which has appeared in the streaming literature in the context of the Count-Min Sketch \cite{cormode2005improved}. As far as we know, our work is the first that looks at sublinear decoding time in the one-bit compressed sensing framework and even achieves less measurements in some cases, contributions that could be regarded as considerable improvements or additions over previous works. We believe that a strong point of our schemes is their simplicity. 

\section{ Preliminaries}

We define the sign function as $\mathrm{sign}(z) = +1$, for $z \geq 0$ and $\mathrm{sign}(z) = -1$ for $z < 0$. For a vector $x$, we define $\mathrm{sign}(x)_i = \mathrm{sign}(x_i)$, for all $i \in [n]$.

Any one-bit compressed sensing scheme is defined by a pair $(\mathcal{D}, Dec)$ where $\mathcal{D}$ is a distribution over $\mathbb{R}^{m \times n}$ and $Dec$ is an algorithm that takes input $\mathrm{sign}(\Phi x)$ for some $x \in \mathbb{R}^n$ and gives back a vector $\hat{x}$. We will refer to $Dec$ either the ``decoder'' or ``decoding procedure''. We use $m$ to denote the number of measurements, and {\em decoding time} refers to the running time of $Dec$. We also define 
$\Sigma_k = \{ x: \| x \|_0 \leq k, \|x \|_2 \leq 1\} $ to be the set of all $k$-sparse vectors contained in the unit $\ell_2$ ball, and $\Sigma^1_k = \{ x: \|x \|_2 = 1, \|x\|_0 \leq k\}$ the set of unit norm vectors with at most $k$ non-zero coordinates. 

For $x \in \mathbb{R}^n$ we denote its support set by $supp(x)$. For a vector $x$,  $head(k)$ denotes the set of its $k$ largest coordinates in magnitude, while $tail(k)$ denotes the set of its $n-k$ smallest coordinates in magnitude. 

For each $S \subset n$, let $x_S \in \mathbb{R}^{|S|}$ denote the signal $x$ restricted to coordinates in $S$. Similarly, for a matrix $M \in \mathbb{R}^{r \times n}$ and each $S \subset n$ let $M_S \in  \mathbb{R}^{ r \times |S|}$ be the matrix $M$ restricted to columns in $S$.

 Given two matrices $A,B$ we denote the Hadamard or entrywise product by $A \odot B$; so the $(i,j)$ entry of $AB$ is $ A_{ij} \cdot B_{ij}$. The row-direct sum $A \biguplus B \in \mathbb{R}^{(r_1 + r_2) \times n}$ of $A,B$ is defined as the vertical concatenation of $A$ and $B$.\newline

\begin{definition}
We say a scheme satisfies the $\delta-\ell_2 / \ell_2$ guarantee for one-bit compressed sensing  if for each $x \in \mathcal{S}^{n-1}$, it estimates a vector $\hat{x}$ such that 
\[ \|x - \hat{x} \|_2^2 \leq C \|x_{tail(k)}\|_2^2 + \delta,\]

where $C$ is an absolute constant.
\end{definition}

We note again that the additive factor of $\delta$ is necessary here, in contrast to linear compressed sensing.

We also give the definition of the tensor product of two matrices. We note that this is not the standard tensor product (or Kronecker product, as usually known) appearing in the literature.

\begin{definition}
 Let $A \in \{0, 1\}^m \times \{0, 1\}^N$ and $A' \in \{0, 1\}^{m'} \times \{0,1\}^N$. The tensor
product $ A \otimes A'$
is an $m m' \times N$ binary matrix with rows indexed by the elements of $[m] \times [m']$ such
that for $i \in [m]$ and $i' \in [m'
]$, the row of $A \otimes A′$
indexed by $(i, i')$ is the coordinate-wise product
of the $i$-th row of $A$ and $i'$-
th row of $A'$.

\end{definition}

In order to proceed, we have to explain the difference between the for-all and the for-each model. Let $P_x$ be the predicate that the sparse recovery scheme returns a vector $\hat{x}$ such that $\| \frac{x}{\|x\|_2} - \frac{\hat{x}}{\| \hat{x} \|_2} \|_2^2 > \delta$, when the matrix $\Phi$ is chosen from the distribution $\mathcal{D}$. Let $p$ be some target probability. In the for-each model the guarantee is that $ \forall x \in \Sigma_k^1,  \mathbb{P}[ P_x ]  \leq p$. In the for-all model the guarantee is that $ \mathbb{P}[ \exists x \in \Sigma_k^1: P_x ] \leq p$. The randomness of the scheme is over the distribution $\mathcal{D}$.

We should review some folklore definitions from Combinatorial Group Testing theory. One can see \cite{HPA}.
\begin{definition}
A $t \times n$ matrix $M$ is $k$-disjunct if for every set $S \subset [n]$ with  $|S| \leq k, \forall j \notin S, \exists i$ such that $M_{i,j}=1$ but $\forall k \in S, M_{i,k} = 0$. In other words, $supp(M_j) - \cup_{l \in S} supp(M_l) \neq \emptyset$.
\end{definition}

\begin{definition}
A $t \times n$ matrix $M$ is $(k,l)$-disjunct if for every two disjoint sets $S,T \subset [n]$ with $|S| \leq k, |T| \leq l$, there exists a row $i$ such that $\forall j \in T, M_{i,j} = 0$, but $\exists j \in S, M_{i,j} =1 $.
\end{definition}

\begin{algorithm}
\caption{Naive Decoding Algorithm}
\label{alg0}
\begin{algorithmic}

\STATE $S \leftarrow  \emptyset$
\FOR {$i \in [n]$ } 
	\IF { exists no negative test where $i$ participates in} 
		\STATE $S \leftarrow S \cup \{i\}$
	\ENDIF
	\ENDFOR
\STATE Output $S$.
\end{algorithmic}
\end{algorithm}

We will make extensive use of the following two lemmas:

\begin{lemma}
Let $M$ be a $k$-disjunct matrix. Then, given $ y = Mx$, the naive decoding algorithm returns a set $S$ such that $S = supp(x)$, i.e. the naive decoding algorithm correctly finds the support of $x$.
\end{lemma}

\begin{lemma}
Let $M$ be a $(k,l)$-disjunct matrix. Then, given $ y = Mx$, the naive decoding algorithm returns a set $S$ such that $supp(x) \subset S$ and $|S| \leq |supp(x)|+l$, i.e. the naive decoding algorithm finds a superset of the support of $x$ with additional $l$ elements.
\end{lemma}

Of course, the two different definitions solve a different problem; the latter one solving a more relaxed version of Group Testing than the former. We will refer to the second version as two-stage group testing, whereas we will refer to the first version just as group testing.

\subsection{Overview of techniques}

For the the $\delta-\ell_2/\ell_2$ guarantee, we use the main idea from \cite{charikar2002finding} and show that essentially a variant of the Count Sketch and the dyadic trick \cite{cormode2005improved} can be implemented using only one-bit measurements. Of course, we cannot approximate the values of a vector $x$, but we can find the coordinates which carry a significant fraction of the $\ell_2$ mass (at least $\frac{1}{\sqrt{10k}}\|x_{tail(k)}\|_2$), something which is sufficient for our case, since the algorithm of \cite{Plan-robust} is used later, in order to approximate the vector $x$.\newline
For the noiseless case, as well as the support recovery algorithm, we use techniques and schemes developed in the context of two-stage group testing. We then show how these schemes can be extended to schemes for finding the support or the superset of the support of a vector $x$, given access only to one-bit measurements.\\
All of our algorithms find either the support or a superset of the support of the vector $x$ (noiseless case) or a set containing the largest $\Oh(k)$ in magnitude coordinates (noisy case) and hence the crucial information needed to approximate $x$. Then, by restricting to the set obtained, we show how the algorithm from \cite{Plan-robust} can give us the desired guarantees.

\section{For-each $\delta-\ell_2 / \ell_2$ One-Bit Compressed Sensing}

In this section we give an algorithm that achieves the $\delta-\ell_2/\ell_2$ guarantee. The algorithm is based on a two-stage approach. The first stage identifies the set $S$ of the `heavy' coordinates of the vector $x$; these coordinates carry most of the $\ell_2$ mass of $x$ and hence the crucial information needed to approximate it. The second stage restricts to the columns indexed by $S$ and runs the convex program of \cite{Plan-robust}. We then show that this suffices for the $\delta-\ell_2/\ell_2$ guarantee.

The most important part of the algorithm, that enables sublinear decoding time, is the procedure that finds the set $S$. We turn our attention to the Count-Sketch and the dyadic trick \cite{charikar2002finding, cormode2005improved}, and show that, with only a constant multiplicative increase in the measurement complexity, we can modify it so that it also works with one-bit measurements.

For the first stage of our algorithm, we show that given access only to the signs of the measurements of the Count Sketch suffices to identify the `heavy' coordinates, and hence the set $S$. However, since we are aiming for sublinear decoding time, we also use the idea of the dyadic trick: we keep multiple sketches corresponding to a hierarhical representation of the vector $x$, which allow us to find the heavy hitters faster than looking at all coordinates of $x$ one by one.

 Unfortunately, this stage does not give us the values of the magnitudes in $S$. That is why we keep another sketch, which is essential for approximating the vector $x_S$. This second sketch which corresponds to the second stage, will be a matrix where each entry is a standard random variable.  Once we have $x$, we can consider only the columns of the sketch corresponding to $S$ and run the algorithm of \cite{Plan-robust} there.\newline

In what follows, we assume that $n$ is a power of $2$. Let $C_{-1},C_0,C_1,C_2$ be large enough constants to be defined later. .We also define $L_l^a = [2^{(a-1)l}, \ldots ,2^{a \cdot l -1} ], \Delta = \frac{1}{C_{-1} k \log n}, \Delta' = \log( \frac{1}{\Delta})$, where $C_{-1}$.

The sensing matrix consists of the vertical concatenation of $\log n$ matrices $E_1,E_2,\ldots, E_{ \log n} , \Phi$. The number of rows of each $E_l$ is $C_0 \cdot C_1 \cdot C_2\cdot k \Delta'$ and the number of rows of $\Phi$ is $\Oh( \delta^{-2} k )$. For $1 \leq l \leq \log n$ let $E_l$ be the $l$-th matrix. $E_l$ consists of submatrices $E_l^1,\ldots,E_l^{C_2 \Delta'}$. Now, $E_l^m$ is inspired and it is very similar to the Count-Sketch matrix of \cite{charikar2002finding}: each matrix $E_l^m$ consists of $C_1$ matrices $E_{l,t}^m$  , $t = 1, \ldots ,C_1$. For clearness, we define the $q$-th row of $E_{l,t}^m$ via its dot product with $x$:

\[  \inprod{ e_q^T E_{l,t}^m , x } = \sum_{a=1}^{\frac{n}{2^l}} \delta_{l,t,q,a}^m \cdot  \sigma_{l,t,q,a}^m \sum_{j \in L_l^a} g_{l,j}^m \cdot x_j,   \]

where $g_{l,j}^m  \sim \mathcal{N} (0,1)$ and $\delta_{l,t,q,a}^m$ are Bernoulli random variables with $\mathbb{E} \delta_{l,t,q,a}^m  = \frac{1}{C_0k}$.  Moreover, $\sum_q \delta_{q,t,l,a}^m = 1$.

In other words, every $E_l$ holds a hierarchical separation of $[n]$ into intervals of length $2^l$. Fix now some $m \in [C_2 \Delta']$. Then, in each $E_{l,t}^m$, every interval is hashed to some bucket and the coordinates inside this interval are combined with standard gaussians. Moreover, every interval is assigned a random sign. The intuition is that with constant probability, we do expect the term $\sum_{j \in L_l^a} g_{l,j}^m \cdot x_j$, to behave roughly like the $l_2$ mass of the interval itself. Then, by keeping the same gaussians, we take $C_1$ such hashing schemes (we refresh only the $\sigma$ and $\delta$ variables). For fixed $m$, this corresponds to each level having this variant of Count-Sketch. Then, for each level, we repeat the whole scheme $C_2 \Delta'$ times, for $m =1, \ldots , C_2 \Delta'$. Note now that across $E_l^m$ for different $m,l$ we use new $g$ variables. The reason we have to make this additional repetition, in contrast to the standard dyadic trick, is that we only have sign information and we cannot use fresh gaussians at every measurement, since this would imply uniformity of the signs of the measurements. In other words, we would roughly see half $+1$ and half $-1$; however, we will heavily exploit the fact that some  measurements will all have the same sign. \newline 

The decoding algorithm processes these intervals in decreasing $l$ for $l=\log n, \ldots, 1$  and keeps a list of intervals at each time (the list is denoted by $S_l$ in the pseudocode). In the beginning of each step $l$, every interval is hashed to $C_0 k $ buckets.  Suppose for a moment, that we have the $\ell_2$  mass of each interval and we hash these values, instead, into $C_0 k$ buckets combined with random signs. If an interval contains a node that is `heavy', then we expect that its $\ell_2$ mass dominates the mass that is hashed into the same bucket from other nodes. Thus, the sign of the sum must be determined by the sign of the `heavy' interval. To overcome the fact that we do not have the $\ell_2$ mass of the interval (since we can only make use of linear measurements) we add a standard random variable in front of every node in the inteval, before hashing. We exploit the aforementioned intuition, along with $2$-stability of the Gaussian distribution , to show that we can identify all `heavy' intervals and that we do not introduce a big number of erroneous intervals (intervals that are not `heavy'). We repeat this hashing scheme $C_1$ times with the same gaussians and try to find the intervals whose sign `frequently' agrees or disagrees with the measurement they participate in. We consider them good. Let this whole hashing scheme called $Scheme~1$. Now, we repeat $Scheme~1$ $C_2 \Delta'$ times with completely fresh randomness. We then find the intervals were consider good at least $\frac{2}{3}C_2 \Delta'$ times and add them to a list. At the end of each step $l$, every interval $L_l^i$ that belongs to the list and is recognised as heavy, is substituted by its two sub-intervals $L_{l-1}^{2i-1},L_{l-1}^{2i}$. \newline

We will need the constant $C_{thr}$, which is defined as the maximum constant such that the following holds:\newline
$\mathbb{P}_{ Y \sim \mathcal{N} (0,1) }[ Y^2 > \frac{1}{C_{thr}}] = \frac{1}{10}.$ \newline

\begin{algorithm}                      
\caption{Recovery of Heavy Hitters from One-Bit Measurements}          
\label{alg1}                           
\begin{algorithmic}                    
	\STATE $l \leftarrow \log n -1$	
	\STATE $S_{\log n} \leftarrow \{ [n]\}$
	\WHILE {  $l > 0 $   } 
	\FOR {each $L_{l+1}^{i}$ in $S_{l+1}$} 
	\STATE $S_{l} \leftarrow L_{l}^{2i -1} \cup L_l^{2i}$.
	\ENDFOR
	\STATE $H_l \leftarrow \emptyset$
		\FOR{ every element $L_l^a$ in $S_l$ }
			\FOR{ $m=1$ to $C_2 \Delta'$}
			\STATE $good \leftarrow 0$
			\FOR {$t=1$ to $C_1$}
			\STATE Let $q$ be the index of the row of the measurement that $L_l^a$ participates in $E_{l,t}^m$.
					\IF{ $\mathrm{sign}(\inprod{y_q}) = \sigma_{l,t,q,a}^m$}
					\STATE $cnt \leftarrow cnt + 1;$
					\ENDIF
			\ENDFOR
			\IF {$cnt > 0.8 C_1 ~or~cnt < 0.2  C_1$} 
				\STATE $good \leftarrow good+1$
			\ENDIF	
		\ENDFOR
		\IF {$good > \frac{2}{3} C_2 \Delta'$}
			\STATE $H_l \leftarrow H_{l} \cup { L_l^i}$
		\ENDIF
	\ENDFOR
		\STATE $S_{l} \leftarrow H_l $		
    \ENDWHILE
    \STATE Output every $x$ in $S_{1}$.	
\end{algorithmic}
\end{algorithm}

\begin{algorithm}
\caption{Decoding given $y = \mathrm{sign}(Ax)$}
\label{alg2}
\begin{algorithmic}

\STATE $S \leftarrow  Algorithm~2()$.
\STATE  $\hat{x} = argmax \inprod{y,A_S z}$ subject to $\|z\|_2 \leq 1, \|z\|_1 \leq \sqrt{k}$.
\STATE Output $\hat{x}$.
\end{algorithmic}
\end{algorithm}

The main result of this section is the following.

\begin{theorem}
Let $A$ be the sensing matrix that is defined as the vertical concatenation of 
$ E_{1},\ldots,E_{\log n}, \Phi$. Let also $Algorithm~3$ be its associated procedure. Then the scheme $(A,Algorithm~3)$ satisfies the $\delta-\ell_2/\ell_2$ guarantee with probability $1 - \Oh(\frac{1}{k})$. Moreover, $A$ has $\Oh(k \log n (\log k + \log \log n) + \delta^{-2} k)$ rows.
\end{theorem}

For the proof of correctness of our algorithm, we need the following definition:

\begin{definition}
For a coordinate $i$ and a level $l$, let $b^l(i)$ be such that $i \in L_l^{b^l(i)}$. If the level $l$ is implicit, we may simplify the notation to $b(i)$.
\end{definition}

When we say that an interval $I$ is considered good by some matrix $E_{l}^m$ with constant probability, we mean that in our pseudocode the variable $\mathrm{cnt}$ is going to be increased with constant probability when we are considering the interval $I$.
We now move with the main lemmas, $Lemma~3$ and $Lemma~4$.

\begin{lemma}
Fix $m,l$. Let also a coordinate $i$ such that $|x_i|^2 > \frac{1}{10k} \|x_{tail(k)}\|_2^2$. Let also be a interval $I$ at the same level. Assume that $L_l^{b(i)} , I \in S_l$. Then, for some absolute constant $c$, the following claims hold:
\begin{itemize}

\item The interval $L_l^{b(i)}$ will be considered good by $E_l^m$ with constant probability.

\item If there are at least $ck$ intervals at the same level $l$ which have greater $\ell_2$ mass than $I$, then, with constant probabiliy, $I$ is not going to be considered good by $E_l^m$.
\end{itemize}
\end{lemma}

\begin{proof}

Fix some measurement $q$, where $L_l^{b(i)}$ participates to. Define the random variable $Y = \sum_{j \in L_l^{b(i)}} g_{l,j}^m x_j$. By $2$-stability of the gaussian distribution we know that
 $Y \sim \mathcal{N} (0, \|x_{L_l^{b(i)}} \|_2)$. This implies that

\[ \mathbb{P} [ Y^2 > \frac{1}{C_{thr}} \sum_{j \in L_l^{b(i)}}  x_j^2 ] = \frac{9}{10} \] 

We now look at the `noise' coming to the measurement $q$. Consider the at most $k$ intervals at level $l$ which contain the largest $k$ in magnitude coordinates of $x$. Since the expected number of these (at most) $k$ intervals participating in $q$ is $\frac{1}{C_0}$, for suitable choice of $C_0$, with probability $9/10$, none of the $k$ heaviest coordinates of $x$ will participate in measurement $q$. From now one, we condition on this event. For the measurement $q$ define $W_q = \sum_{a \in [\frac{n}{2^l}]- \{b(i)\}} \delta_{l,t,q,a}^m \cdot  \sum_{j \in L_l^a}  x_j^2$. Observe now that 

\[ \mathbb{E}[W_q]  = \frac{\|x_{tail(k)}\|_2^2}{C_0 k}. \] 

Now, we will argue that for all $C_1$ different $q$ where $L_l^{b(i)}$ participates in, the expected squared noise will be close to  $\mathbb{E}[W_q]$. Indeed, consider the $C_1$ different values 
\[ \sum_{a \in [\frac{n}{2^l}]-\{b(i)\}} \delta_{l,t,q,a}^m \cdot  \sigma_{l,t,q,a}^m \sum_{j \in L_l^a} g_{l,j}^m \cdot x_j, \]

for all different $q$.

Then, with probability at least $\frac{9}{10}$, $max_q \{ \frac{1}{W_q}| \sum_{a \in [\frac{n}{2^l}]- \{b(i)\}} \delta_{l,t,q,a}^m \cdot  \sigma_{l,t,q,a}^m \sum_{j \in L_l^a} g_{l,j}^m \cdot x_j| \} < C_5$, where $C_5$ is an absolute constant. Hence, for any $q$,  the expected squared `noise' $ \mathbb{E}[ (\sum_{a- \in [\frac{n}{2^l}] -\{ b(i)\}} \delta_{l,t,q,a}^m \cdot  \sigma_{l,t,q,a}^m \sum_{j \in L_l^a} g_{l,j}^m \cdot x_j)^2]$ will be at most $\frac{C_5}{C_0 k} \|x_{tail(k)} \|_2^2$.

 Let us fix again $q$ now. By Markov's inequality, the squared noise in $q$ will exceed $\frac{10 C_5}{C_0 k} \|x_{tail(k)}\|_2^2$ with probability at most $\frac{1}{10}$. On the other side, $Y$ will be at least $ \frac{1}{10 k \cdot C_{thr}} \|x_{tail(k)}\|_2^2$. This means that if $C_0$ is sufficiently large, then $\frac{10 C_5}{C_0} < \frac{1}{10 C_{thr}}$ and the mass in $q$ will be dominated by the mass of $x_{L_{b(i)}^l}$, with probability $1 - ( \frac{1}{10} + \frac{1}{10} + \frac{1}{10}) = 0.7$. Hence, with probability at least $p_1 + (1-p_1)\frac{1}{2}= 0.85$, the sign of the measurement will agree with $\sigma_{L_l^{b(i)}}$ or will disagree with it. \newline

We now turn our attention to the second bullet of the lemma. Let again $q$ be the row of a measurement of which $I$ participates in. For large enough $c$, with probability at least $\frac{1}{10}$, we will have $ C_5 C_{thr}'$ elements with $\ell_2$ mass bigger than $\|x_{I}\|_2$ that will participate to this measurement. This means that the squared $\ell_2$ mass of this measurement will be at least $C_5 C_{thr}'$ times more than $\|x_{I}\|_2^2$. Again, by the same argument as before, the `noise' will dominate the measurement of this measurement, with probability at least $\frac{9}{10}$. Hence, with probability at least $1-p_2$, where $p_2 = \frac{1}{10} + \frac{1}{10}$ the sign of the measurement will be a either $+1$ or $-1$ with the same probability. This means $1-p_2$ of the time we expect the sign of $y_{\rho}$ to be uniformly at random.

Given the above two claims, we see that, using $E_{l,t}^m$  we can recognise either of the next two cases with constant probability: 
if an interval contains an element $i$ such that $|x_i|_2^2 > \frac{1}{10k} \|x_{tail(k)}\|_2^2$ or if there are at least $ck$ intervals at the same level with larger mass. This completes the proof of the lemma.

\end{proof}

\begin{lemma}
\label{poutsa}
For any $ 1\leq l \leq \log n$, the following holds for the set $S_l$: \\
\begin{itemize}
\item If $ |x_i|^2 > \frac{1}{10k} \|x_{tail(k)}\|_2^2$  and $i \in L_l^{b(i)}$, then $L_l^{b(i)} \in S_l$.
\item $|S_l| = \Oh( k )$.
\end{itemize}
\end{lemma}

\begin{proof}

For simplicity, we will introduct some additional definitions. We will refer to any interval that contains a node $i$ such that $ |x_i|^2 > \frac{1}{10k} \|x_{tail(k)}\|_2^2$, as a type $1$ interval. For any level $l$, we say that an interval at level $l$ is of type $2$, if there exist at least  $c k$ intervals at the same level that have greater mass than this interval. \\

We now proceed by induction on the number of levels. The base case is trivial. 
We focus on some level $l$ and assume that the induction hypothesis holds for all previous levels $l$. We prove the first bullet. Let $i$ be a coordinate such that $|x_i|_2^2> \frac{1}{10k}\|x_{tail(k)}\|_2^2$. By the induction hypothesis we get that $L_{l+1}^{b^{l+1}(i)} \in S_{l+1}$ and hence $L_l^{2b^l(i)} \in S_l$. From $Lemma~2$ we know that for any $l,m$, $E_{l}^m$ will clasify a type-1 interval as good with constant probability $>\frac{1}{2}$. Moreover, it will not classify any type-2 interval as good, again with constant probability $> \frac{1}{2}$. This implies that after repeating the same scheme $C_2 \Delta' = C_2 \log(\frac{1}{\Delta})$ times, we will know, with probability at least $1 - \Delta$, if a specific interval is a type-1 interval or not or if it is a type interval or not $2$. By setting $\Delta = \Theta( \frac{1}{k \log n})$ so that we can take a union-bound over all possible intervals we might consider ($\Oh(k)$ at each of the $\log n$ levels), we can guarantee that every type-1 interval will be added to $H_l$, while any type-2 interval will not be added to $H_l$. This implies that at any step we have at most $2ck$ intervals in $S_l$ and no type-2 intervals will beadded to $H_l$ -and hence at $S_l$- at the end of each outer iteration of the algorithm.

\end{proof}

The following theorem is proved in \cite{Plan-robust}.
\begin{theorem}
Let $\Phi$ be a random $m \times n$ matrix, with each entry being a standard gaussian, and all entries are independent. Let $x \in \Sigma_k^1$ and $y=\mathrm{sign}(\Phi x + v)$, where $v \sim \mathcal{N}(0,I)$. Then, the convex program\\
	
	\[ z = argmax \inprod{y,\Phi z}\]
	\[ s.t.\]
	 \[\|z\|_2 \leq 1, \|z\|_1 \leq \sqrt{k}, \]

outputs a $\hat{x}$ such that $\|\hat{x} - x\|_2^2 \leq \delta$ with probability $1 - e^{-\Omega(k \log (n/k))}$, as long as $m = \Omega( \delta^{-2}k \log (n/k))$.

\end{theorem}

We move with the proof of Theorem $1$, which is a combinatorion of $Lemma~4$ and the aforementioned theorem.\\
\begin{proof}

By running $Algorithm~2$, we get a set $S$ that satisfies the guarantees of $Lemma~4$. Then, $y = \mathrm{sign}(\Phi x) = \mathrm{sign}(\Phi x_S + \Phi x_{[n]-S}) = \mathrm{sign}(\Phi x_s + v)$, where $v$ is a vector each entry of which follows normal distribution with variance $\| x_{[n] - S} \|_2^2 \leq 1$. The number of rows needed by $Theorem~2$ equals $\Omega(\delta^{-2} k \log(ck/k)) = \Omega(\delta^{-2}k)$. Hence, since the number of rows of $\Phi$ is a large enough constant times $\delta^{-2}k$, Theorem $2$ applies and the convex program outputs a vector $\hat{x}$ such that $\| \hat{x} - x_S \|_2^2 \leq \delta$. Since every coordinate $i$ with $|x_i|^2 \geq \frac{1}{10k} \|x_{tail(k)} \|_2^2$ is contained in $S$, \[ \|x_{[n] - S}\|_2^2 \leq \|x_{tail(k)}\|_2^2 + ck \frac{1}{10 k} \|x_{tail(k)}\|_2^2 = (1 + \frac{c}{10}) \|x_{tail(k)}\|_2^2.\] This implies that $\|x - \hat{x} \|_2^2 \leq (1+\frac{c}{10}) \| x_{tail(k)} \|_2^2 + \delta$, as desired.
\end{proof}

\section{Fast decoding for noiseless signals}

\subsection{For-each Group Testing}

We start by solving the for-each version of the two-stage group testing problem. This will be a crucial building block for our algorithm.

We proceed by showing the construction of this matrix. We will denote the matrix by $A'_n$. This is a standard construction; however, we redo the analysis because previous analyses in the literature we are aware of were looking at the for-all version of the problem and hence were obtaining worse bounds.\\
First, we will randomly construct a class of matrices and then we will combine them recursively, to obtain our desired matrix $A_n'$. For any $m \leq n$, employ the following construction.
Let $q = 4k$ and $ d =  5\log k + {\lceil} \frac{\log \log_k n}{k} {\rceil}$.
Let $C_{id}: [q] \rightarrow \{0,1\}^{q}$ be the identity code,namely the collection of all $4k$ standard basis vectors in $\mathbb{R}^{q}$. Let $C_r : [n] \rightarrow [q]^d$ be a random code.  Consider the concatenation of the aforementioned codes, $C_{id} \circ C_r$. Let $A_m$ be the matrix that has columns the codewords of the concatenated code $C_{id} \circ C_{r}$. In other words, $A_m$ is the vertical concatenation of $d$ smaller submatrices, each one of dimension $4k \times n$, where every column of every submatrix has a $1$ at a random position and the other entries on that column are $0$.

For the construction of our desired matrix $A_n'$, we closely follow \cite{HPA}, using recursion.\newline

The construction goes as follows:

For any $m \leq n$, let $A_{\sqrt{m}}^{(R)}$ be the $s(k,\sqrt{m}) \times m$ matrix where the $i$-th column is identical to the $j$-th column of $A_{m}$ such that the first $\frac{1}{2} \log m$ bits of $i$ is $j$. Similarly, $A_{\sqrt{m}}^{(C)}$ is the $s(k,\sqrt{n}) \times m$ where the last $\frac{1}{2} \log n$ bits of $i$ is $j$.  Now, we set $A' = A_{\sqrt{m}}^{(R)} \biguplus A_{\sqrt{m}}^{(C)} \biguplus A_m$. Let $S_1, S_2$ be  the sets returned by the algorithm when the decoding procedure is applied on $(A_{\sqrt{m}}^{(R)} x), (A_{\sqrt{m}}^{(C)} x)$.  Given any $k$-sparse vector $x$, any candidate of the support of $supp(x)$  must be in the intersection of some element in $S_1$ and some element in $S_2$. This means that there are at most $4k^2$ possible elements. Then, using the initial $A_n$ one can find the support, by restricting to these $4k^2$ elements. 

We can summarize the routine constructing the matrix $A_n'$ as follows:
\begin{itemize}
\item If $m \leq k^2$, return $A_m$ 
\item Else, return $A_{\sqrt{m}}^{(R)} \biguplus A_{\sqrt{m}}^{(C)} \biguplus A_m$.
\end{itemize}

\begin{lemma}
Let $A_m$ be a matrix  with $m$ rows, constructed at some point in the recursion procesure. Let $S$ be any subset of the columns of $A_m$, with at most $4k^2$ elements. Then $A_m^S$ ( $A_m$ restricted on the columns of S) is $(k,k)$-disjunct with probability at least $1 - 4^{-5k\log k - \log \log_k n}$.
\end{lemma}

\begin{proof}

Let $T,T'$ be two sets of columns of matrix $A_m^S$, such that $T \cap T' = \emptyset$ and  $|T| = |T'| = k$. The columns in $T,T'$ correspond to two sets of codewords. For each position $i$, let $T_i$ and $S_i$ denote the set of symbols which the codewords in $T$ and $T'$ have at that position, respectively. We restrict ourselves to a block of consecutive $q$ rows, which corresponds to a symbol of the outer code. Then, when restricted to this block of rows, the union of columns in $T'$ is contained in the union of columns in $T$ with constant probability $\frac{1}{4}$. Hence, with probability at most $(\frac{1}{4})^{5k \log k + \log \log_k n}$, the union of columns in $T$ will contain the union of columns in $T'$. Now, using $Lemma~1$,  we get that $A_m^S$ is $(k,k)$-disjunct with error probability at most ${4k^2 \choose k} {4k^2 - k \choose k}  4^{-5k\log k - \log \log_k n} \leq (4k^2)^k\cdot (4k^2)^{k} 4^{-5k \log k - \log \log_k n} \leq 2^{ 4k \log k + 4k } 2^{- 10k \log k - 2\log \log_k n} \leq 4^{-5k\log k - \log \log_k n}$

\end{proof}

\begin{lemma}
Given $A'_n x$, we can find in time $\Oh( (k^2 (\log k +  \log \log_k n)) \log_k n)$ a superset of the defected items with at most $2k$ items. The number of rows of $A'$ is $\Oh( k \log n +  \log_k n \log \log_k n)$.

\end{lemma}

\begin{proof}

We assume that $ n = k^{2^T}$ for some positive integer $T$. We will drop this assumption later.
Observe that the recursion corresponds to a binary tree of height $T = \log \log_k n$, because at each time we take the square root of $k$, till we reach $k^2$. At each node of the tree, we run the naive decoding algorithm on the matrix $A_m$ of that node, restricted on $4k^2$ elements. By $Lemma~4$, the naive algorithm succeeds with probability $4^{-5k \log k -\log \log_k n}$. Since we have $2^T = \log_k n$ nodes in the tree, by a union bound, at each node of the tree the naive algorithm will succeed. Hence, at the end of the execution of decoding algorithm, we are going to get a set with the desired guarantees.\\

The decoding time is easy, for at its node of the tree, once we have the result from its children, we spend $\Oh( k^2 (\log k  +  \log \log_k n))$ to check for all $4k^2$ elements if they are defective. This happens because the column sparsity of each matrix at each node is $d = \Oh( \log k + \log \log_k n)$, and hence, when restricting ourselves to a specific coordinate, the naive decoding algorithm needs to check only $d$ for that element. By summing over all nodes of the tree, we need $\Oh( (k^2( \log k +  \log \log_k n) \log_k n)$ time in total. Moreover, since we have $\log_k n$ nodes in the tree, we have that the total number of rows of $A$ is $\Oh( k \log k +  \log \log_k n ) \cdot \log_k n = \Oh( k \log n +  \log_k n \cdot \log \log_k n)$. \\

To remove the assumption $ n = k^{2^T}$, take $i^*$ such that $k^{2^{i^*-1}} < n \leq k^{2^{i^*}}$. It is easy to check that removing any $k^{2^{i^*}} - n$ arbitrary columns from $A_{k^{2^{i^*}}}$, gives a matrix with the desired guarantees, if we treat the removed columns as corresponding to negative items.

\end{proof}

\subsection{For-each Compressed Sensing from One-Bit Measurements}

Let $G_1 \in \mathbb{R}^{r \times n}, G_2 \in \mathbb{R}^{r' \times n}$ be two gaussian matrices, i.e. each entry is a random variable coming from the normal distribution, where $r' = \Oh(\delta^{-2}k)$ and $r$ equals the number of rows of $A_n$. Our measurement matrix is going to be $ \Phi = (A'_n \odot G1 ) \biguplus ( -A'_n \odot G_1 ) \biguplus G_2 $. This matrix has $m = 2r + r'$ rows. We remind that $A'_n$ is the matrix constructed in the previous section. 	The following theorem holds:\\

\begin{theorem}

Let $x \in \Sigma_k^1$. Given $y = \mathrm{sign}(\Phi x)$, the decoder can find a $k$ sparse vector $\hat{x}$ such that $\| x - \hat{x} \|_2^2 < \delta$ in time $\Oh( k^2 (\log k + \log \log_k n) \log \log_k n)+poly(\delta^{-2} k)$, with error probability at most $2\exp( -C\delta^{-1} k) + 4^{ -k \log k} + 4^{-C \delta^{-2} k}$. The number of rows of $\Phi$ is $\Oh(k \log n + \log_k n \log \log_k n + \delta^{-2}k)$. 

\end{theorem}

\begin{proof}

Let $T = supp(x)$ and fix any index $i$ that corresponds to a row of $A_n'$. If  $supp(e_i^T A_n') \cap T = \emptyset$ then $\mathrm{sign}( \inprod{e_i^T(A_n' \odot G_1),x} = )1$ and 
 $\mathrm{sign}( \inprod{e_i^T(-A_n' \odot G_1),x} =1 $. If $supp(e_i^T A_n') \cap T \neq \emptyset$ we have with probability $1$ that $\inprod{e_i^T(A_n' \odot G_1),x} \neq 0$. This means that \\
$\mathrm{sign}(\inprod{e_i^T(-A_n' \odot G_1),x}) \cdot \mathrm{sign}( \inprod{e_i^T(A_n' \odot G_1),x}) = -1$, so we can recognize that $T$ and the support of that row have a non-empty intersection.  This is translated in the combinatorial group testing framework: we know for the given queries if a given query has a common intersection with $T$. By the group testing decoding algorithm we can find a set $S \subseteq [n]$ with at most $2k$ elements such that $T \subseteq S$ in time $\Oh(k^2(\log k + \log \log_k n)\log \log_k n$. Observe now that $G_2^S x_S = G_2x$. Since we have the sign values of $G_2 x$, we can set up the following convex optimization program:

\[  max \inprod{y,G_2^Sx}, s.t. x \in \Sigma_{k} \]

As mentioned before, this algorithm is analyzed in \cite{Plan-robust} where the authors show that the optimal solution $\hat{x}$ satisfies ctor $ \|\hat{x} - x \|_2^2 \leq \delta$. This completes the proof.

\end{proof}

\subsection{For-all One-Bit Compressed Sensing and Support Recovery}

Our approach is again the same: first we find a superset of the support, and then we run the convex programming algorithm of \cite{Plan-robust} restricted on this set. The crucial idea is again the recovery of this superset from $\Oh(k^2 \log (n/k) \log \log_k n)$ measurements and $poly(k, \log n)$ decoding time.  For the construction, let $A$ be any $(k,k)$-disjunct matrix that allows $poly(k,\log n)$ decoding time. Such a matrix is guaranteed by \cite{HPA}. The construction is similar to the one presented in setction 4.1 but with different parameters and the analysis been done for the for-all version. Then, the same recursion trick can be used to bring the decoding time down to $poly(k, \log n)$, with a multiplicative factor of $\log \log_k n$. The intuition behind the double $\log$ factor in the number of measurements is that the total number of rows of all matrices at every level of the recursion tree is $\Oh(k \log n)$. 

Now, let $V$ be a $k \times n$ a Vandermorde matrix with all columns different, and consider the matrix $ A'' = A \otimes V$. We claim that the following modification of the decoding algorithm recognises a superset of the support of $x \in \Sigma_k$, given $\mathrm{sign}( -(A \otimes V)x)$ and $\mathrm{sign}( (A \otimes V)x)$: Suppose that at some step the the decoding algorithm checks the result of the $j$-th row of $A$ dotted with $x$ and checks if it is $0$ or not. The modified decoding algorithm checks if $((e_j^TA ) \otimes V)x$ is the zero vector. If this is the case, then it proceeds as the initial decoding algorithm would do if the answer was $0$, otherwise it proceeds as if the answer was $1$. We will refer to this algorithm as the modified decoding algorithm.

\begin{algorithm}
\caption{Modified Decoding Algorithm}
\begin{algorithmic}

\STATE $S \leftarrow  \emptyset$
\FOR {$i \in [n]$ } 
	\STATE $ belongs \leftarrow true$
	\FOR { all tests $j$ of $A$ where $i$ participates in }
		\IF { $(A''x)_{(j-1)*k + m} = 0, \forall m \in [k]$}
			\STATE $belongs \leftarrow false$
		\ENDIF
	\ENDFOR
		\IF { $belongs = true$ }
		\STATE $S \leftarrow S \cup \{i\}$
		\ENDIF
\ENDFOR
\STATE Output $S$.
\end{algorithmic}
\end{algorithm}

If we were using $A$ as our sensing matrix and all values of $x$ were $0$ or $1$, then the decoding algorithm would give the desired result. In our case $A$ would not suffice as our sensing matrix, since we would get false negatives. Taking the tensor product with $G$, we can circumvent this. The lemma below makes this intuition concrete.

\begin{lemma}
Given $\mathrm{sign}(A''x) \biguplus  \mathrm{sign}(-A''x)$ for some $x \in \Sigma_k^1$, the decoder can find a set $S$, such that $supp(x) \subseteq S$ and $|S| \leq 2k$.
\end{lemma}

\begin{proof}
Let $r_i$ be the $i$-th row of $A$. Let $a_i$ be the support of the $i$-th row of $A$ and let $V^{(i)}$ be the restriction of $V$ on the columns indexed by any $k$ element subset that contains $a_i \cap supp(x)$. Every $k \times k$ submatrix of $V$ is invertible since its determinant is non-zero with probability $1$. Since we look at $k$-sparse vectors, if all entries of $(r_i \otimes V)x$ were zero, this would mean that $V^{(i)}$ has  a non-trivial kernal, which is a contradiction. Hence, we can correctly perform each test that the initial algorithm performs, which implies that the modified decoding algorithm is correct.
\end{proof}

\begin{theorem}

Let $\Phi = A'' \biguplus - A'' \biguplus G_2$, where $G_2$ is a $\Oh( \delta^{-2} k \log( n/k)) \times n$ random gaussian matrix. Then for all $x \in \Sigma_k^1$, given $y = \mathrm{sign}(\Phi x)$, one can find $\hat{x} \in \mathbb{R}^n$ such that $\| x - \hat{x} \|_2^2 \leq \delta$ in $poly(k,\log n)$ time. The matrix $\Phi$ has $\Oh( k^2 \log (n/k) \log \log_k n + \delta^{-2} k \log(n/k) ) $ rows.	

\end{theorem}

\begin{proof}
Similar to before, from $A'' \biguplus -A''$ the decoder can find a superset $S$ of the support of $x$ in time $poly(k,\log n)$. Then, using Theorem $1$ he can find the desired $\hat{x}$ in time $poly(k)$, by restricting to the columns of $G_2$ indexed by $S$, and then running the algorithm of \cite{Plan-robust}.

\end{proof}

We also explain how to improve the main algorithm for support identification that appears in \cite{support}. There, the authors give a scheme with $\Oh(nk\log n)$ decoding time and $\Oh( k^3 \log n)$ measurements. Here, we give a recovery scheme that is not only exponentially faster in terms on the dependence of $n$, but also computable in polynomial time. For this, we will use a result of Indyk, Rudra and Ngo \cite{indyk}:

\begin{theorem}
There exists a $t \times n$ $k$-disjunct matrix, with $t = \Oh(k^2 \log n)$, which is decodable in $poly(t)$ time. Moreover, the matrix can be computed in $\tilde{\Oh} (n t)$ time.

\end{theorem}

Equipped with this theorem, we are ready to state and prove our result on support recovery.

\begin{theorem}
There is a $\Oh(k^3 \log n) \times n $ matrix $B'$, such that for all $x \in \Sigma_k$, given $y = B'x$, we can find in time $poly(k, \log n)$ a set $S$ such that $S = supp(x)$. Moreover, $B'$ can be computed in time $\Oh(nk^3 \log n)$.
\end{theorem}

\begin{proof}

Let $B$ be the matrix guaranteed by $Theorem~5$. We combine, as before, the $k$-disjunct matrix with our Vandermorde matrix, to get $B' = B \otimes V$. Clearly, $B'$ has $\Oh(k^3 \log n)$ rows and the time to construct $B'$ equals $\tilde{\Oh}(n k^2 \log n) + \Oh(k \cdot nk^2 \log n) = \Oh (n k^3     \log n)$. Furthermore, the modified decoding algorithm finds the support of any vector $k$-sparse vector $x$, since each test corresponding to the matrix $B$, can be implemented using $k$ tests from $B'$. This proves the statement of the theorem.

\end{proof}

\subparagraph{Acknowledgements}

The author would like to thank Ely Porat for pointing to him \cite{HPA}, as well as Jelani Nelson for helpful discussions.






\bibliographystyle{alpha}
\bibliography{bibio}

\end{document}